%% file: domains_randomvariables.tex
\documentclass[11pt]{entcs}
\usepackage{entcsmacro,latexsym,amssymb,amsmath}
\input macros
\renewcommand{\Pr}[1]{{\textsf{Prob}\,#1}}
\newcommand{\SPr}[1]{{\textsf{SProb}\,#1}}

\newcommand{\dom}{\text{dom}\,}
\newcommand{\V}{{\mathbb V}}

\renewcommand{\to}{\longrightarrow}
\newcommand{\ran}{\text{ran}\,} 
\newcommand{\wid}[1]{\widehat{#1}}

\def\firstheadline{}
\begin{document}
\begin{frontmatter}
  \title{Domain Theory and Random Variables}
\author{Michael Mislove}
\address{Department of Computer Science\\Tulane University, New Orleans, LA 70118}
\begin{abstract}
The aim of this paper is to establish a theory of random variables on domains. Domain theory is a fundamental component of theoretical computer science, providing mathematical models of computational processes. Random variables are the mainstay of probability theory. Since computational models increasingly involve probabilistic aspects, it's only natural to explore the relationship between these two areas. Our main results show how to cast results about random variables using a domain-theoretic approach. The pay-off is an extension of the results from probability measures to sub-probability measures. We also use our approach to extend the class of domains for which we can classify the domain structure of the space of sub-probability measures. 
\end{abstract}
\end{frontmatter}
\section{Introduction}
This paper draws its impetus from a line of work whose goal is to develop a domain-theoretic approach to random variables.\footnote{In the probability theory literature, random variables are measurable maps from a probability space that take values in the reals, while \emph{random elements} are measurable mappings from a probability space to an arbitrary measure space. Here we use ``random variables" to denote either.} The original motivation was to use random variables to devise models for probabilistic computation that don't suffer from the well-known problems of the probabilistic power domain~\cite{jungtix}, a program that began about 10 years ago, and recently has seen some notable successes -- more on that below. 

In this paper, we shift the focus from constructing monads for probabilistic choice to laying a foundation for a theory of random variables using domains. We show that an important result from the theory of random variables can be recast in the setting of domain theory, where measurable maps can then be approximated by Scott-continuous maps. The result in question is Skorohod's Theorem~\cite{skorohod}, one of the basic results in stochastic process theory. In its simplest form, this theorem states that any Borel probability measure on a Polish space $P$ can be realized as the law for a random variable $X\colon [0,1]\to P$. That is, if $\mu$ is a Borel measure on a Polish space $X$ and if $\lambda$ denotes Lebesgue measure on the unit interval, then there is a measurable map $X\colon [0,1]\to X$ satisfying $\mu = X\, \lambda$, the push forward of $\lambda$ under $X$. Furthermore, if $\mu_n\to_w \mu$ in $\Pr{P}$ in the weak topology, then the random variables $X, X_n\colon [0,1]\to P$ with laws $\mu$ and $\mu_n$, respectively, satisfy $X_n\to X$ almost surely. This result allows one to replace arguments about the convergence of measures in the weak topology with arguments about almost sure convergence of measurable maps from the unit interval to Polish spaces. It led Skorohod to develop the theory of c\`adl\`ag functions that play a prominent role in the analysis of stochastic processes. 

Our main results are inspired by Skorohod's Theorem. Each of our results generalizes from probability measures to sub-probability measures, which are more commonly used in domain theory. 

Our first result extending Skorohod's Theorem involves two new ingredients. First, in moving to the domain setting, we develop an approach to proving Skorohod's theorem in which the Cantor set, $C$, replaces the unit interval, and a domain $D$ is the target space. The role of Lebesgue measure is played by $\nu_C$, Haar measure on $C$ regarded as a countable product of two-point groups. We also show that Skorohod's Theorem with the unit interval and Lebesgue measure follows as a corollary of our approach. 

We introduce the Cantor set into the discussion because it offers a ready-made computational model in the form of the Cantor tree, $CT$ -- the full rooted binary tree whose set of maximal elements is isomorphic to the Cantor set.  Achieving our results requires the Cantor tree: if we tried our approach using just the Cantor set, which is a chain in the natural order, then all monotone images also would be chains, which would limit the result. But the Cantor tree proves to be just the right structure to generalize to arbitrary countably-based domains as images. This brings up the second new component of our approach: the use of the transport numbers between simple measures that are fundamental to the Splitting Lemma for simple measures on a domain. They allow us to define a sequence of Scott-continuous maps from the Cantor tree to a target domain that approximate a given measure on the domain. 

To continue the discussion, we need some notation: we realize the Cantor tree, $CT \simeq \{0,1\}^*\cup \{0,1\}^\omega$ as the set of finite and infinite words over $\{0,1\}$.  If we endow $CT$ with the prefix order, then $\P_C(CT)$, the convex power domain of $CT$ is a coherent domain. If we denote by $C_n$ the set of words of length $n$, then we let $\nu_n$ denote normalized counting measure on $C_n$, which is Haar measure when $C_n\simeq 2^n$ is regarded as a finite group. Moreover, we have $\nu_n = \pi_n\, \nu_C$, where $\pi_n\colon C\to C_n$ is the canonical projection. In fact, $\nu_C = \sqcup_n \nu_n$ in \textsf{SProb}$\, CT$, the family of sub-probability measures on $CT$, regarded as Scott-continuous valuations over $CT$ and ordered pointwise.  If $D$ is a domain, then we let $[CT\rightharpoonup D]$ denote the family of Scott-continuous maps $f\colon A\to D$, where $A$ is a Lawson-closed antichain in $CT$, with the order $f\leq g$ iff $\dom f\subseteq \da \dom g$ and $f\circ \pi_{\dom f,\dom g}\leq g$ where all components are defined.  Our first main result is the following:
\medbreak
\noindent\textbf{Theorem 1.} \emph{(Skorohod's Theorem for Domains) Let $D$ be a countably-based coherent domain, and let $\{\mu_n\}_n\in \textsf{SProb}\,D$ be sequence of Borel sub-probability measures satisfying $\lim_n \mu_n  = \mu \in \textsf{SProb}\,D$, where the limit is taken in the Lawson topology. Then there are Scott-continuous maps $f, f_n\colon CT\rightharpoonup D$ satisfying $f\vert_C\, \nu_C = \mu, f_n\vert_C\, \nu_n = \mu_n$ for each $n$, and $\lim_n f_n \longrightarrow f$ in the Lawson topology on $[CT\rightharpoonup D]$. }
\medbreak
Skorohod's Theorem is a corollary of Theorem 1 as follows. Any Polish space $P$ has a \emph{computational model,} a countably-based bounded complete domain $D$ for which $P$ is homeomorphic to the set Max$\,P_D$ of maximal elements endowed with the relative Scott topology. In fact, Max$\,P_D$ is a $G_\delta$, hence a Borel subset of $P_D$. The last piece is provided by the fact that the canonical surjection $\phi\colon C\to [0,1]$ of the Cantor set onto the unit interval preserves all sups and infs, and so it has a lower adjoint $j\colon [0,1]\to C$ preserving all suprema. Following $j$ by the maps provided in Theorem 1 then yields Skorohod's original result. 

Our second theorem is a special case of the discussion above, when the Polish space $P$ is actually totally ordered. In this case, we abandon our indirect approach using the Cantor tree, and instead take a direct approach to considering mappings from the unit interval, but also restricting to the case that $P$ is a complete chain.  This allows us to prove the following result using direct, domain-theoretic arguments:
\medbreak
\noindent\textbf{Theorem 3.} If $D$ is a complete chain with $\perp$ as the only compact element, then $\textsf{SProb}\,D$, the family of sub-probability measures on $D$, and $\Pr{D}$, the family of probability measures on $D$, are continuous lattices. 
\medbreak
This result significantly expands our knowledge of the domain structure of the family of sub-probability measures on a domain $D$. Indeed, up to this point, the only domains $D$ for which the domain structure of $\textsf{SProb}\,D$ is known are a tree, $T$, for which $\textsf{SProb}\,T\in$ \textsf{BCD}, the category of bounded complete domains, or a finite reverse tree $T^{op}$, in which case $\textsf{SProb}\,T^{rev}$ is in \RB\cite{jungtix}.

\subsection{Related Work}
Previous work that is related to our results include Edalat's extensive history of results devising domain-theoretic approaches to topics such as integration theory~\cite{edalat1}, stochastic processes~\cite{edalat2}, dynamical systems and fractals~\cite{edalat3}, and Brownian motion~\cite{bilokon}. His development with Heckmann of the formal ball model~\cite{edalat5} provided an approach tailored to modeling metric spaces and Lipschitz maps using domain theory. The concept of a computational model emerged in Edalat's work on domain models of spaces arising in real analysis using the domain of compact subsets under reverse inclusion, where the target space arises as the set of maximal elements. The first paper formally presenting such a model was~\cite{edalat3}, where a domain model for locally compact second countable spaces was given. That paper presents a range of applications of the approach, including dynamical systems, iterated function systems and fractal, a computational model for classical measure theory  on locally compact spaces, and a computational generalization of Riemann integration.  Related work led to the formal ball model~\cite{edalat5} which was tailor-made for modeling metric spaces and Lipschitz functions. Further discussion of these developments occurs in our discussion of Polish spaces in Section~\ref{sec:polish} below.

Other related work concerns the program to develop random variable models of probabilistic computational processes. This began with~\cite{mislove-icalp}, a paper that provided a domain model for finite random variables. Further efforts had limited success until recently. The model proposed in~\cite{g-l-v-lics} turned out to be flawed, as was initially shown in~\cite{misl-anat1,misl-anat2}.   But inspired by ideas from~\cite{g-l-v-lics}, Barker~\cite{barker} devised a monad of random variables that gives an abstract model for randomized algorithms. This line of research was initiated by Scott~\cite{scott-stoch}, who showed how the $\P(\Nat)$ model of the lambda calculus could be extended naturally to support probabilistic choice with the aid of a random variable $X\colon [0,1]\to \P(\Nat)$. Barker's results abstract Scott's approach by providing a model of randomized PCF that adds a version of probabilistic choice based on random variables. Finally, the author has devised another monad based on random variables~\cite{mislove-newmon} that supports settings in which processes, such as those representing honest participants in a crypto-protocol, for instance, have access to distinct sources of randomness, something that Barker's monad does not support. It is notable that both of these monads leave important Cartesian closed categories of domains invariant -- in particular, the category \BCD of bounded complete domains, as well as the CCC \RB of retracts of bifinite domains invariant, and each enjoys a distributive law with respect to at least one of nondeterminism monads.\\[1ex] 

The rest of the paper is as follows. In the next section, we review the material we need from a number of areas, domain theory, topology, and probability theory. Section 3 develops results about mappings from the Cantor tree to the space $\SPr{D}$ of sub-probability measures on a countably-based coherent domain $D$. Section 4 contains the main results of the paper, by first recalling the development of Polish spaces as computational models, and then presenting the main theorems.  Section 5 summarizes what's been proved, and discusses future work. 

\section{Background}\label{sec:background}
In this section we present the background material we need for our main results. 
\subsection{Domains}\label{subsec:domains}
Our results rely fundamentally on domain theory. Most of the results that we quote below all can be found in \cite{abrjung} or \cite{comp}; we give specific references for those that are not. 

To start, a \emph{poset} is a partially ordered set. A poset is \emph{directed complete} if each of its directed subsets has a least upper bound, where a subset $S$ is \emph{directed} if each finite subset of $S$ has an upper bound in $S$. A directed complete partial order is called a \emph{dcpo}. 
The relevant maps between dcpos are the monotone maps that also preserve suprema of directed sets; these maps are usually called \emph{Scott continuous}. 

From a purely topological perspective, a subset $U\subseteq P$ of a poset is \emph{Scott open} if (i) $U = \ua U \equiv \{ x\in P\mid (\exists u\in U) \ u\leq x\}$ is an upper set, and (ii) if $\sup S\in U$ implies $S\cap U\not=\emptyset$ for each directed subset $S\subseteq P$. It is routine to show that the family of Scott-open sets forms a topology on any poset; this topology satisfies $\da x \equiv \{y\in P\mid y\leq x\} = \overline{\{x\}}$ is the closure of a point, so the Scott topology is always $T_0$, but it is $T_1$ iff $P$ is a flat poset. In any case, a mapping between dcpos is Scott continuous in the order-theoretic sense iff it is a monotone map that is continuous with respect to the Scott topologies on its domain and range.  We let \DCPO denote the category of dcpos and Scott-continuous maps; \DCPO is a Cartesian closed category.

If $P$ is a dcpo, and $x,y\in P$, then \emph{$x$ approximates $y$} iff for every directed set $S\subseteq P$, if $y\leq \sup S$, then there is some $s\in S$ with $x\leq s$. In this case, we write $x\ll y$ and we let $\Da y = \{x\in P\mid x\ll y\}$. A \emph{basis} for a poset $P$ is a family $B\subseteq P$ satisfying $\Da y\cap B$ is directed and $y = \sup (\Da y\cap B)$ for each $y\in P$. A \emph{continuous poset} is one that has a basis, and a dcpo $P$ is a \emph{domain} if $P$ is a continuous dcpo. An element $k\in P$ is \emph{compact} if $x\ll x$, and $P$ is \emph{algebraic} if $KP = \{ k\in P\mid k\ll k\}$ forms a basis. Domains are sober spaces in the Scott topology. 

We let \DOM denote that category of domains and Scott continuous maps; this is a full subcategory of \textsf{DCPO}, but it is not Cartesian closed. Nevertheless, \DOM has several Cartesian closed full subcategories. Two of particular interest to us are the full subcategory \SDOM of Scott domains, and \BCD its continuous analog. Precisely, a \emph{Scott domain} is an algebraic domain for which $KP$ is countable and that also satisfies the property that every non-empty subset of $P$ has a greatest lower bound. An equivalent statement to the last condition is that every subset of $P$ with an upper bound has a least upper bound. A domain is \emph{bounded complete} if it also satisfies this last property that every non-empty subset has a greatest lower bound; \BCD denotes the category of bounded complete domains and Scott-continuous maps. 

Domains admit a Hausdorff refinement of the Scott topology which will play a role in our work. The \emph{weak lower topology} on $P$ has the sets of the form if $O = P\setminus \ua F$ as a basis, where $F\subset P$ is a finite subset. The \emph{Lawson topology} on a domain $P$ is the common refinement of the Scott- and weak lower topologies on $P$. This topology has the family 
$$\{ U\setminus\!\ua F\mid U\ \text{Scott open}\ \&\ F\subseteq P\ \text{finite}\}$$
as a basis. The Lawson topology on a domain is always Hausdorff. 

A domain is \emph{coherent} if its Lawson topology is compact. We denote the closure of a subset $X\subseteq P$ of a coherent domain in the Lawson topology by $\overline{X}^\Lambda$.

Two examples of coherent domains that we need are the Cantor tree and the unit interval. If $A$ is a finite set, then $A^\infty$ denotes the set of finite and infinite words over $A$. In this case, the family $\{\ua k\setminus \ua F\mid k\in KA^\infty\ \&\ F\subseteq KA^\infty\ \text{finite}\}$ is a base for the Lawson topology. The fact that $\ua k$ is clopen in the Lawson topology for each compact element $k$ implies that the Lawson topology on a coherent algebraic domain is totally disconnected. Of particular interest is $A = \{0,1\}$, in which case $\{0,1\}^\infty$ is the \emph{binary Cantor tree} whose set of maximal elements $\{0,1\}^\omega$ is the Cantor set in the relative Scott (= relative Lawson) topology. 
 
The other example is the unit interval $[0,1]$, where $x\ll y$ iff $x = 0$ or $x < y$. The Scott topology on the $[0,1]$ has basic open sets $[0,1]$ together with $\Ua x = (x,1]$ for $x \in (0,1)$. Since \DOM has finite products, $[0,1]^n$ is a domain in the product order, where $x\ll y$ iff $x_i\ll y_i$ for each $i$; a basis of Scott-open sets is formed by the sets $\Ua x$ for $x\in [0,1]^n$ (this last is true in any domain). 
 
 The Lawson topology on [0,1] has basic open sets $(x,1]\setminus [y,1]$ for $x < y$ -- i.e., sets of the form $(x,y)$ for $x < y$, which is the usual topology. Then, the Lawson topology on $[0,1]^n$ is the  product topology from the usual topology on $[0,1]$. 
 
 Since $[0,1]$ has a least element, the same results apply for any power of $[0,1]$, where $x\ll y$ in $[0,1]^J$ iff $x_j = 0$ for almost all $j\in J$, and $x_j\ll y_j$ for all $j\in J$. Thus, every power of $[0,1]$ is a coherent domain. 

We note that all of these examples -- including the last one if $J$ is countable -- are \emph{countably based} domains. That is, each has a countable basis. A result that plays an important role for us is the following:
\begin{lemma}\label{lem:cntbase}
If $D$ is a countably based domain. then every $x\in D$ is the supremum of a countable chain $\{x_n\}_{n\in\Nat}$ with $x_n\ll x$ for each $n$.
\end{lemma}
\begin{proof}
Since $D$ has a countable base, there is a countable directed set $B\subseteq \Da x$ with $\sqcup D = x$. If we enumerate $B = \{b_0,b_1,\ldots\}$, then we define the desired sequence $x_n$ as follows: $x_0 = b_0$, and if $x_0\ll x_1\ll \ldots\ll x_n$ have been chosen from $B$, then we choose $x_{n+1}\in B$ with $b_i\ll x_{n+1}$ for each $i\leq n$ and $x_n\ll x_{n+1}$. This extends the sequence, and then a standard maximality argument shows we can choose an countable sequence $x_n$ with $x_n\ll x_{n+1}\ll x$ for each $n$. Finally, $x = \sqcup_n b_n \leq \sqcup_n x_n$, since $b_n\ll x_{n+1}$ for each $n$, but $x_n\ll x$ for each $n$ implies $\sqcup_n x_n\leq x$. 
\end{proof}
We also need some basic results about Galois adjunctions (cf.~Section 0-3 of \cite{comp}) in the context of complete lattices. If $L$ and $M$ are complete lattices, a \emph{Galois adjunction} is a pair of mappings $g\colon L\to M$ and $f\colon M\to L$ satisfying $f\circ g\leq 1_L$ and $g\circ f\geq 1_M$. In this case, $f$ is the \emph{lower adjoint}, and $g$ is the \emph{upper adjoint}. Lower adjoints preserve all suprema, and upper adjoints preserve all infima. In fact, each mapping $f$ between complete lattices that preserves all suprema is a lower adjoint; its upper adjoint $g$ is defined by $g(y) = \sup f^{-1}(\da y)$. Dually, each mapping $g$ preserving all infima is an upper adjoint; its lower adjoint $f$ is defined by $f(x) = \inf g^{-1}(\ua x)$. The cumulative distribution function of a probability measure on $[0,1]$ and its upper adjoint given in the introduction are examples we'll find relevant. 

\subsection{Subprobability measures on domains and the probabilistic power domain}
The probability-theoretic approach to measures on a complete metric space $X$ starts by considering the Borel $\sigma$-field $\B$ generated by the open sets, and then defines a sub-probability measure as a non-negative, countably additive set function $\mu\colon\B\to \Re_+$ satisfying $\mu(X)\leq 1$.\footnote{Most texts confine the discussion to probability measures, but the results we need are valid for sub-probability measures. We provide proofs for the results we need below.} Each such measure $\mu$ then defines an integral $\int f d\mu$ for  $f\in C_C(X, \Re)$, the Banach space of continuous functions of compact support, for example by approximating using simple functions. The sub-probability measures \textsf{SProb}$\,X$ then can be endowed with the \emph{weak topology}, in which $\mu_n\longrightarrow_w \mu$ iff $\mu_n(f) \longrightarrow \mu(f)$ for each $f\in C_C(X,\Re)$. 

There also is a functional-analytic approach, which starts with the continuous bounded functions $C_b(X,\Re)$ on a locally compact Hausdorff space $X$, and then considers the \emph{dual space} of continuous linear functionals $\phi\colon C_b(X,\Re)\to \Re$. The Riesz Representation Theorem shows there is an isomorphism $\M(X) \simeq C_b(X,\Re)^*$ between the space of measures on $X$ and the dual space of $C_b(X,\Re)$. We then can endow $\M(X)$ with the \emph{weak *-topology:} $\mu_n \longrightarrow \mu$ iff $\int f d\mu_n\to \int f d\mu$ for all $f\in C_b(X,\Re)$.
A functional $\phi$ is \emph{positive} if $\phi(f)\geq 0$ for all $f\in C_b(X,\Re_+)$, and then the isomorphism above restricts to one between the sub-probability measures \textsf{SProb}$\, X$ and the positive linear functionals $\phi$ with norm $||\phi||\leq 1$. 

These two approaches coincide -- and the weak- and weak$^*$-topologies agree -- when $X$ is a compact metric space. In particular, they agree for a countably-based coherent domain $D$ endowed with the Lawson topology. 

Domain theory traditionally takes yet a third approach to sub-probability measures, one that emphasizes the order structure. In this approach, the sub-probability measures over a domain $D$ are viewed as \emph{continuous valuations:} mappings $\mu\colon \O(D)\to [0,1]$ from the family of Scott-open sets to the interval satisfying: 
\begin{itemize}
\item (Strictness) $\mu(\emptyset) = 0$,
\item (Modularity) $\mu(U\cup V) + \mu(U\cap V) = \mu(U) + \mu(V)$, for $U, V\in \O(D)$,
\item (Scott continuity) If $\{U_i\}_{i\in I}\subseteq \O(D)$ is directed, then $\mu(\bigcup_i U_i) = \sup_i \mu(U_i)$.
\end{itemize}
Valuations are ordered pointwise: $\mu\leq \nu$ iff $\mu(U)\leq \nu(U)$ for all $U\in\O(D)$. 
We denote the set of valuations over a domain $D$ with this order by $\V D$. 

It is straightforward to show that each Borel sub-probability measure restricts to a unique Scott-continuous valuation on the Scott-open sets. The converse, that each Scott-continuous valuation on a dcpo extends to a unique Borel sub-probability measure was shown by Alvarez-Manilla, Edalat and Saheb-Djorhomi~\cite{alvman}. 

Linking the order-theoretic approach to $\V D$ and the approaches to \textsf{SProb}$\, D$ outlined above relies on the next result. We recall that a \emph{simple sub-probability measure} on a space $X$ is a finite convex sum $\sum_{x\in F} r_x\delta_x$, where $F\subseteq X$ is finite, $r_x\geq 0$ for each $x\in F$, and $\sum_{x\in F} r_x\leq 1$. The following is called the \emph{Splitting Lemma:}
\begin{theorem}\label{thm:split} (Splitting Lemma~\cite{jones}) 
Let $D$ be a domain and let $\mu = \sum_{x\in F} r_x\delta_x$ and $\nu = \sum_{y\in G} s_y\delta_y$ be simple sub-probability measures on $D$. Then the following are equivalent:
\begin{enumerate}
\item $\mu \leq \nu \in \V D$,
\item There is a family $\{t_{x,y}\}_{\langle x,y\rangle\in F\times G}\subseteq [0,1]$ satisfying:
\begin{itemize}
\item $r_x = \sum_{y\in G} t_{x,y}$ for each $x\in F$,
\item $\sum_{x\in F} t_{x,y} \leq s_y$ for each $y\in G$,
\item $t_{x,y} > 0\ \Rightarrow\ x\leq y$.
\end{itemize}
\end{enumerate}
Moreover, $\mu\ll \nu \in \V D$ iff  (i) $\sum_{x\in F} r_x < s_y$ for each $y\in G$ and (ii) $t_{x,y}$ satisfies $t_{x,y} > 0$ implies $x\ll y\in D$ for each $x\in F, y\in G$. 
\end{theorem}
This result can be used to show that, given a basis $B$ for $D$, the family $\{ \sum_{x\in F} r_x\delta_x\mid F\subseteq B, \sum_{x\in F} r_x < 1\}$ forms a basis for $\V D$; in particular, each sub-probability measure is the directed supremum of simple measures way-below it, so $\V D$ is a domain if $D$ is one. Moreover, Jung and Tix~\cite{jungtix} showed that $\V D$ is a coherent domain if $D$ is. 

Our interest is in countably-based coherent domains, in which case we can refine the Splitting Lemma~\ref{thm:split} and Lemma~\ref{lem:cntbase}; here $Dyad$ denotes the family of dyadic rationals in the unit interval:
\begin{corollary}\label{cor:split}
If $D$ is a coherent domain with countable basis $B_D$, then $\V\, D$ is a countably-based coherent domain with basis\\[1ex]
\centerline{ $B_{\V D} = \{ \sum_{x\in F} r_x\delta_x\mid F\subseteq B_D\ \text{finite}\ \&\ r_x\in Dyad\ \forall x\in F\}$.}\\[1ex] 
Moreover, 
if $\sum_{x\in F} r_x\delta_x\leq \sum_{y\in G} s_y\delta_y\in B_{\V D}$, then the family $\{t_{x,y}\}_{(x,y)\in F\times G}$ of transport numbers from the Splitting Lemma~\ref{thm:split} can be chosen to satisfy $t_{x,y}\in Dyad$ for all $(x,y)\in F\times G$.

Finally, 
each $\mu\in \V D$ is the supremum of a countable chain $\mu_n\in B_{\V D}$. 
\end{corollary}
\begin{proof}
It is shown in~\cite{jungtix} that $\V D$ is coherent if $D$ is, and the Splitting Lemma~\ref{thm:split} implies $B_{\V D}$ is a basis for $\V D$. 

We next outline the proof of the second point -- that the transport numbers $t_{x,y}$ between comparable simple measures all belong to $Dyad$ if the coefficients of the measures do. This follows from the proof of the Splitting Lemma~\ref{thm:split} as presented in~\cite{jones}: That proof is an application of the Max Flow -- Min Cut Theorem to the directed graph $G = (E,N)$ which has a ``source node," $\perp$, connected by an outgoing edge of weight $r_x$ to each ``node" $x\in F$, a  ``sink node," $\top$, with an incoming edge of weight $s_y$ from each element $y\in G$,  and edges from $x\in F$ to $y\in G$ of large weight (say, $1$), if $x\leq y$. 

A flow is an assignment $f\colon E \to \Re_+$ of non-negative numbers to each edge so that $f(uv) \leq c(uv)$ for nodes $u, v$, where $c(uv)$ is the weight as defined above, and satisfying $\sum_u f(uv) = \sum_t f(vt)$ for each node $t\not=\perp, \top$. The value of a flow $f$ is $val f = \sum_{u} f(\perp\! u)$, the total amount of flow out of $\perp$ using $f$.  A cut is a partition of $N = S\stackrel{\cdot}{\cup} T$ with $\perp\in S$ and $\top\in T$. The value of a flow across the cut $T$ is $\sum_{(u,v) \in S\times T\,\cap\, E} f(uv)$. 

The Max Flow -- Min Cut Theorem asserts that the maximum flow on a directed graph is equal to the minimum cut. It is proved by applying the Ford--Fulkerson Algorithm~\cite{bolla}. The algorithm starts by assigning the minimum flow $f(uv) = 0$ for all edges $(u,v)\in E$, and then iterates a process of selecting a path from $\perp$ to $\top$, calculating the residual capacity of each edge in the path, defining a residual graph $G_f$, augmenting the paths in $G_f$ to include additional flow, and then iterating. The result of the algorithm is the set of flows along edges across the cut, which are the transport numbers $t_{x,y}$ in our case. Since the calculations of new edge weights involve only arithmetic operations, and since the dyadic rationals form a subsemigroup of $\Re_+$, the resulting transport numbers $t_{x,y}$ are dyadic rationals if the coefficients of the input distributions are dyadic. 

The final assertion follows from the fact that $B_{\V D}$ is a basis, by an application of Lemma~\ref{lem:cntbase}.
\end{proof}

There remains a question of the order structure on $\textsf{SProb}\, D$ that arises from $\V\, D$. 
To clarify this point, we first 
recall that the real numbers, $\Re$, are a continuous poset whose Scott topology has the intervals $(a,\infty)$ as a basis, and whose Lawson topology is the usual topology. 
\begin{proposition}\label{weak-law}
Let $D$ be a coherent domain, and let $\mu, \nu$ be sub-probability measures on $D$. Then the following conditions are equivalent:
\begin{enumerate}
\item $\mu\leq \nu\in \V D$. 
\item For each Scott-continuous map $f\colon D\to \Re_+$, $\int f d\mu \leq \int f d\nu$.
\item For each monotone Lawson-continuous $f\colon D\to \Re_+$, $\int f\, d\mu \leq \int f\, d\nu$.
\end{enumerate}
\end{proposition}
\begin{proof}
We show the result for simple measures, which then implies it holds for all measures since $\V D$ is a domain -- so its partial order is (topologically) closed -- in which the simple measures are dense.
  
So, suppose $\mu = \sum_{x\in F} r_x\delta_x $ and $\nu = \sum_{y\in G} s_y \delta_y$ are simple measures on $D$. 
\medbreak
\noindent \textbf{(i) implies (ii):} Suppose that $\mu \leq \nu \in \V D$. If $f\colon D\to \Re_+$, then $\int f d\mu = \sum_{x\in F} r_x\cdot f(x)$ and $\int f d\nu = \sum_{y\in G} s_y\cdot f(y)$. Since $\mu\leq \nu$, there are $t_{x,y}\in [0,1]$ guaranteed by the Splitting Lemma~\ref{thm:split}, and so
\begin{eqnarray*}
\int f d\mu & = & \sum_{x\in F} r_x\cdot f(x) = \sum_{x\in F}\sum_{y\in G} t_{x,y}\cdot f(x)\\
& \leq & \sum_{x\in F}\sum_{y\in G} t_{x,y}\cdot f(y) \leq \sum_{y\in G} s_y\cdot f(y) = \int f d\nu,
\end{eqnarray*}
where the first inequality follows from the facts that $t_{x,y} > 0$ implies $x\leq y$ and $f$ is monotone. This shows (i) implies (ii).
\medbreak
\noindent \textbf{(ii) implies (iii):} Since monotone Lawson continuous maps are Scott continuous, this is obvious.
\medbreak
\noindent \textbf{(iii) implies (i):} Let $U$ be a Scott-open subset of $D$, and let $H = (F\cup G)\setminus U$. Using the facts that $D$ is coherent, so its Lawson topology is compact Hausdorff, and that $H$ is finite, we define a family $\{ U_d\mid d\in Dyad\}$ of Scott-open upper sets indexed by $Dyad$, the dyadic numbers in $[0,1]$ as follows: We let $U_0 = D\setminus \da H, U_1 = U$, and for $d < d'$, we recursively choose $U_d\supseteq \overline{U_{d'}}^\Lambda$, the Lawson-closure of $U_{d'}$. Then define a mapping\\[1ex]
\centerline{$f\colon D\to [0,1]$ by $f(x) = 0$ if $x\in \da H$, and otherwise $f(x) = \inf \{ d\mid x\in U_d\}$.}\\[1ex] 
Since the family $\{U_d\}$ consists of Scott-open sets satisfying $U_d\supseteq \overline{U_{d'}}^\Lambda$ for $d < d'$, this mapping is monotone, and the standard Urysohn Lemma argument (cf.\ Theorem 33.1~\cite{munkres}) shows it is Lawson continuous. So, $\int f d\mu\leq \int f d\nu$ by assumption. 

Since $\mu$ and $\nu$ are simple, $\int f d\mu = \sum_{x\in F\setminus H} r_x\cdot f(x)$ and $\int f d\nu = \sum_{y\in G\setminus H} s_y\cdot f(y)$. By construction, $(F\cup G)\setminus H \subseteq U = U_1$, so\\[1ex] 
$\sum_{x\in F\setminus H} r_x\cdot f(x) = \sum_{x\in F\setminus H} r_x = \mu(U)$, and
$\sum_{y\in G\setminus H} s_y\cdot f(y) = \sum_{y\in G\setminus H} s_y = \nu(U)$,\\[1ex]  and so 
$\mu(U) = \int f d\mu \leq \int f d\nu = \nu(U)$, as required.
\end{proof}
\begin{remark} Proposition~\ref{weak-law} offers added insight into the relationship between \textsf{SProb}$\,D$ and $\V D$ for a countably-based coherent domain $D$, by showing that the domain order from $\V D$ can be defined directly on \textsf{SProb}$\, D$ using the maps from $D$ to $\Re_+$. 
\end{remark}

\section{Domain Mappings from the Cantor Tree}\label{sec:cantor}
The \emph{Cantor tree} is the family $CT = \{0,1\}^*\cup \{0,1\}^\omega$ of finite and infinite words over $\{0,1\}$ in the prefix order. Equivalently, $CT$ is the full rooted binary tree which is directed complete, and since it is countably based, this means it is closed under suprema of increasing chains. $CT$ will play the role of the unit interval in our approach to generalizing Skorohod's Theorem to the domain setting. For that, we need some preliminary definitions.  

 An \emph{antichain} is a non-empty subset $A\subseteq CT$ satisfying $a,b \in A$ implies $a$ and $b$ do not compare in the prefix order. An extensive study of Lawson-closed antichains in $CT$ and of measures whose (Lawson) support is such an antichain are given in~\cite{misl-anat1}. The key idea in that work was to associate to a Lawson-closed antichain, $A$ its Scott closure, which turns out to be $\da A$ because $CT$ is a tree. 

Our interest here is somewhat different. The results obtained in~\cite{misl-anat1} were in the context of probability measures, and we want to extend the treatment in~\cite{misl-anat1} to include sub-probability measures, since they arise naturally in the current context.  Our aim also is to define mappings from all of $CT$ to the target domain $D$, rather than simply defining them from the Lawson-support of a particular measure. The mappings we seek ultimately can then be realized as restrictions to the Cantor set of maximal elements of $CT$ of mappings defined on all of $CT$.   Accomplishing these goals will be accomplished by first defining mappings on finite antichains of compact elements in $CT$, and then extending them canonically to all of $CT$. 

\begin{notation}\label{not:nota1}
For the next result, we need to establish some notation.
\begin{enumerate}
\item We let $C_n \simeq 2^n$ be the set of $n$-bit words in $CT$, which forms an antichain. Recall that there is a well-defined retraction mapping $\pi_n\colon CT\to \da C_n$ from the Cantor set onto $\da C_n$; this mapping sends each element of $CT$ to its largest prefix in $\da C_n$.  In addition, if $m \leq n$, then there is a map $\pi_{m,n}\colon C_n \to C_m$ that sends each $n$-bit word to its $m$-bit prefix. 

\item We can restrict the projection $\pi_n$ to $C$, the Cantor set of maximal elements in $CT$, and its image is then $C_n$. This projection has a corresponding embedding $\iota_n\colon C_n\to C$ sending an $n$-bit word to the infinite word all of whose coordinates $m> n$ are $0$. Then $\pi_n\circ \iota_n = 1_{C_n}$ and $ \iota_n\circ\pi_n \leq 1_C$. 
If $C, C'\subseteq CT$ are Lawson-closed antichains with $C\subseteq \da C'$, then there is a canonical partial mapping $\pi_{C,C'}\colon C'\rightharpoonup C$ sending each element of $C'$ to the unique element of $C$ below it, iff there is some element of $C$ below it.  This mapping is continuous in the relative Scott topologies on its domain and $C$. 
\item If $D$ is a domain, then we let $[CT\rightharpoonup D]$ denote the family functions $f\colon C\to D$ where $C\subseteq CT$ is a Lawson-closed antichain and $f$ is continuous in the relative Scott topology inherited from $CT$. 

If $f,g\in [CT\rightharpoonup D]$, then we define $f\leq g$ iff $\dom f\subseteq \da \dom g$ and $f\circ \pi_{\dom f,\dom g} \leq g$ where $\pi_{\dom f,\dom g}$ is defined.

\item If $C_n$ is the set of $n$-bit words, we order $C_n$ with the lexicographic order. Then each dyadic rational $r\in [0,1]$ that can be expressed with denominator $2^n$ can also expressed as an interval in $C_n$, namely, from $0$ to $r$. Moreover, each sequence of such dyadics, $r_1,\ldots, r_k$ whose sum is at most $1$ can be expressed as successive intervals, $r_1 = [0,\ldots, r_1], r_2 = [r_1+1,\ldots, r_2]$, etc. We denote each of these subintervals by $[r_i]$. 
\end{enumerate}
\end{notation}
\begin{proposition}\label{prop:dommap1}
Let $D$ be a bounded complete domain with countable basis $B_D$, and for each $m\in\Nat$, let $C_m$ denote the antichain in $CT$ of $m$-bit words and $\nu_m$ normalized counting measure on $C_m$. 

If $\mu_1\leq \mu_2\leq \cdots\leq \mu_n\leq \cdots$ is an increasing sequence of simple sub-probability measures on $D$ whose coefficients are dyadic rationals, then there is a corresponding sequence $m_1< m_2 <\cdots$ and mappings $f_i\colon C_{m_i}\rightharpoonup B_D$ satisfying $i< j$ implies $f_i\leq f_j$ and $f_i\, \nu_{m_i} = \mu_i$. 
\end{proposition}
\begin{proof} Let $\mu_1\leq \cdots\leq \mu_n\leq \cdots$ be a chain of simple sub-probability measures on $D$, and assume $\mu_i = \sum_{x\in F_i} r_x\delta_x$ with $r_x\in Dyad$ for each $x\in F_i$.  We show by induction that every finite initial chain $\mu_1\leq \cdots\leq \mu_n$ has a corresponding family $f_i\colon C_{m_i}\to B_D$ satisfying $f_i\leq f_{i+1}$ and $f_i\,\nu_{m_i} = \mu_i$ for each $i< n$. An appeal to maximality then proves the result.\\[1ex]
\noindent \textbf{Basis Step:} To begin, recall that $\mu_1\leq \mu_2$ means there are transport numbers $t_{x,y}\in \Re_+$ satisfying the conditions that $r_x = \sum_{y\in F_2} t_{x,y}$ and $\sum_{x\in F_1} t_{x,y} \leq s_y$ for each $x\in F_1, y\in F_2$. Since $r_x, s_y\in Dyad$ for all $x,y$, it follows from Corollary~\ref{cor:split} that the transport numbers $t_{x,y}$ also are dyadic rationals. This implies we can choose $m_1 < m_2$ so that:
\begin{itemize}
\item Every $r_x$ can be expressed as a dyadic with denominator $2^{m_1}$, and
\item Every $s_y$ and every $t_{x,y}$ can be expressed as a dyadic with denominator $2^{m_2}$.
\end{itemize}
We then define $f_1\colon C_{m_1}\to B_D$ by $f_1([r_x]) = x$, for $x\in F_1$,\footnote{To avoid notational clutter, we assume the elements of $F_1$ are given in some total order, and that order is used to enumerate the intervals $[r_x]$ in $C_{m_1}$.}   using the notation from Notation~\ref{not:nota1}. That $f_1\, \nu_{m_1} = \mu_1$ follows from the fact that $f_1(\sum_{c\in C_{m_1}} \delta_{c}) = f_1(\sum_{c\in C_{m_1}} \delta_{f_1(c)}) = \sum_{x\in F_1} r_x\delta_{x} = \mu_1$. 

Defining $f_2$ takes a bit more work, because the tree structure on $CT$ complicates the allocation of all the transport numbers $t_{x,y}$ for $x$ fixed, as $y$ varies over $F_2$. To simply things, for a fixed $x\in F_1$, we let $t_x$ denote the sequence of transport numbers $t_{x,y}$ for $y\in F_2$. We also let $u_y = s_y - \sum_{x\in F_1} t_{x,y}$, for each $y\in F_2$; this represents the portion of $s_y$ not needed to accommodate any of the mass from $\mu_1$. 

Next, we endow $F_1\times F_2$ with the lexicographic order, based on the implicit order we assumed on each component. Then we enumerate $C_{m_2}$ as the sequence of intervals $\{[t_{x,y}]\mid (x,y)\in F_1\times F_2\}$ in lexicographic order, followed by the intervals $\{[u_y]\mid y\in F_2\}$ in the order on $F_2$, and then the final interval $[w_{n_2}]$, the final subinterval of $C_{m_2}$ not needed for any mass in $\mu_2$. 

We then define $f_2\colon C_{m_2}\rightharpoonup B_D$ by $f_2([t_{x,y}])  =f_2([u_y]) = y$ for each $x\in F_1$ and $y\in F_2$, and we leave $f_2$ undefined on $[w_{n_2}]$. A simple calculation shows that $f_2\,\nu_{m_2} = \mu_2$: the mass in $s_y$ consists of $\sum_{x\in F_1} [t_{x,y}]$ that is transported from the $r_x$s, and the remaining mass in $s_y$ is $[u_y]$. 

To show that $f_1\leq f_2$, note that our use of the lexicographic order implies that, for each $x\in F_1$, the sequence of intervals $\{[t_{x,y}]\mid y\in F_2\}$ is in the upper set in $CT$ of the interval $[r_x]$, and since $r_x = \sum_{y\in F_2} t_{x,y}$, this sequence of intervals exhausts the intersection of the upper set of $[r_x]$ with $C_{m_2}$. This means the sequence $[t_x] = \{ [t_{x,y}]\mid y\in F_2\}$ represents exactly the mass $r_x$ transported from $C_{m_1}$ up to $C_{m_2}$. From this it follows that, if $\pi_{m_1,m_2}\colon C_{m_2}\to C_{m_1}$ is the natural projection in $CT$, then $c\in C_{m_2}$ satisfies $\pi_{m_1,m_2}(c) \in [r_x]\subseteq C_{m_1}$ for some $x\in F_1$ implies $c\in [t_x]\subseteq C_{m_2}$. This implies $ f_1\circ \pi_{m_1,m_2}\leq f_2$ where both are defined, which is the definition of $f_1\leq f_2$. \\[1ex]
\noindent\textbf{Inductive Step:}  For the inductive step, we assume we are given $\mu_1\leq \cdots \leq \mu_{n+1}$, and that we have defined $m_1 < m_2<\cdots < m_n$, antichains $C_{m_i}\simeq 2^{m_i}$ and functions $f_i\colon C_{m_i}\rightharpoonup B_D$ with $i\leq j\ \Rightarrow\ f_i\leq f_j$ and $f_i\,\nu_{m_i} = \mu_i$ for $1\leq i < j \leq n$. 

We also assume that $\mu_{n-1} = \sum_{x\in F_{n-1}} r_x\delta_x$,  $\mu_n = \sum_{y\in F_n} s_y\delta_y$,  and $\mu_{n+1} = \sum_{z\in F_{n+1}} t_z\delta_z$. We also assume we are given a partition
$C_{m_n} = \{[a_1],\ldots, [a_r],[u]\}$ into subintervals, with each interval $[a_e]$ partitioned into subintervals $[a_e] = \{[a_{e,y}]\mid y\in F_{n}\}$  satisfying $ \sum_{e\leq r} |[a_{e,y}]|\leq s_y$
for each $y\in F_{n}$, and where $|[u]| = 2^{m_k} - \sum_{e\leq r} |[a_e]|$ represents the final subinterval of $C_{m_n}$ not needed for the total mass in $\mu_{n-1}$. We can also decompose $[u] = \{[u_y]\mid y\in F_n\}\cup \{[w_n]\}$ into subintervals, where $|[u_y]| = s_y - \sum_{e\leq r} |[a_{e,y}]|$, which is the amount of mass at $y$ not needed to accommodate incoming mass from $\mu_{n-1}$, and $[w_n]$ is the final interval representing the remaining ``mass" in $C_n$ after the mass in $\mu_n$ is accounted for. 

Since $\mu_n\leq \mu_{n+1}$,  Corollary~\ref{cor:split} asserts there are transport numbers $u_{y,z}\in \Re_+$ satisfying $s_y = \sum_{z\in F_{n+1}} u_{y,z}$ and $\sum_{y\in F_n} u_{y,z}\leq t_z$ for each $y\leq z$, and $\sum_{z\in F_{n+1}} u_{y,z} = s_y = \sum_{e\leq r} |[a_{e,y}]| + |[u_y]|$. Using the larger denominator $2^{m_{k+1}}$, we decompose each summand of $s_y = \sum_{z\in F_{n+1}} u_{y,z}$ as $[a_{e,y}] = \sum_{z\in F_{n+1}} [b_{(e,y),z}]$ and $[u_y] = \sum_{z\in F_{n+1}} [b_{y,z}]$. This gives us a partition
\begin{eqnarray*}
C_{m_{n+1}} &=& \{[b_{(e,y),z}]\mid 1\leq e\leq r, y\in F_n, z\in F_{n+1}\} \cup \{ [b_{y,z}]\mid y\in F_n, z\in F_{n+1}\}\\
& & \qquad \cup \{[u_z]\mid z\in F_{n+1}\}\cup \{[w_{n+1}]\},
\end{eqnarray*}
where $|[u_z]| = s_z - \sum_{e,y} |[b_{(e,y),z}]|$, and $|[w_{n+1}]| = 2^{m_{n+1}} - \sum_z s_z$. 

We define $f_{n+1}\colon C_{m_{n+1}}\rightharpoonup B_D$ by $f_{n+1}([b_{(e,y),z}]) = z = f_{n+1}([u_z])$. By construction, $\dom f_{n+1} = (\cup_{(e,y),z} \{[b_{(e,y),z}]\}) \cup (\cup_z \{[u_z]\})$, and $f_{n+1}\,\nu_{n+1} = \mu_{n+1}$. It is also clear from the construction that $\pi_{\dom f_n,\dom f_{n+1}}(c) \in [a_{e,y}]$ iff $c\in [b_{(e,y),z}]$ for some $z\in F_{n+1}$, and $\pi_{\dom f_n,\dom f_{n+1}}(c) \in [u_y]$ iff $c\in [b_{y,z}]$ for some $z\in F_{n+1}$. This implies that $f_n\circ \pi_{\dom f_n,\dom f_{n+1}} \leq f_{n+1}$, as required. 

This shows that we can construct the required sequence of antichains $C_{m_n}$ and partial maps $f_n$ for each $n$, and then the standard maximality argument shows that there is such a sequence that works simultaneously for all $n$. 
\end{proof}


For our next result, we need some information about the weak topology on $\textsf{SProb}\, D$. The result we need follows from the Portmanteau Theorem~(cf., e.g.,~\cite{billings}), and a proof can be found as Corollaries 15 and 16 in~\cite{vanbreug}:

\begin{theorem}\label{thm:topvsweak} Let $D$ be a countably based coherent domain endowed with the Borel $\sigma$-algebra. Then the weak topology on $\textsf{SProb}\,D$ is the same as the Lawson topology on $\textsf{SProb}\,D$ when viewed as a family of valuations. 

Moreover, for a family $\mu_n,\mu\in\textsf{SProb}\,D$, the following are equivalent:
\begin{enumerate}
\item $\mu_n\to_w\mu$
\item Both of the following hold:
\begin{itemize}
\item $\lim\sup_n \mu_n(E)\leq \mu(E)$ for all finitely generated upper sets $E\subseteq D$.
\item $\lim\inf_n \mu_n(O)\geq \mu(O)$ for all Scott-open sets $O\subseteq D$.
\end{itemize}
\item $\lim\inf_n \mu_n(O)\geq \mu(O)$ for all Lawson-open sets $O\subseteq D$.
\end{enumerate}
\end{theorem}

\begin{theorem}\label{thm:mapping}
Let $D$ be a countably-based coherent domain, let $C$ denote the Cantor set of maximal elements in $CT$, the Cantor tree, and let $\nu_C$ denote Haar measure on $C$ viewed as a countable product of two-point groups. If $\mu \in \textsf{SProb}(D)$, then there is a  measurable partial map $f_\mu\colon C\rightharpoonup D$ satisfying $f\, \nu_C = \mu$. 
\end{theorem}
\begin{proof}
Using Corollary~\ref{cor:split}, there is an increasing sequence $\mu_1\ll  \cdots\ll\mu_n\ll \cdots$ of simple measures in $\textsf{SProb}(D)$ with dyadic rational coefficients satisfying $\sup_n \mu_n = \mu$.  Then Proposition~\ref{prop:dommap1} implies there is a sequence $C_{m_n}\simeq 2^{m_n}$ of antichains in $K\, CT$ and a sequence $f_n\colon C_{m_n}\rightharpoonup B_D$ satisfying $f_n\, \nu_{m_n} = \mu_n$ and $n \leq n'$ implies $f_n\circ \pi_{m_n,m_{n'}} \leq f_{n'}$ where defined. If $\pi_n\colon C\to C_{m_n}$ is the natural projection, then $\pi_{m_n}\, \nu_C = \nu_{m_n}$, and so  $f_n\, \nu_n = f_n \, (\pi_{m_n}\, \nu_C) = f_n\circ \pi_{m_n}\, \nu_C$, again where $f_n$ and $f_n\circ \pi_{m_n}$ are defined. 

In particular, if we let $C_n' = \iota_n(\dom f_n)\subseteq C$, then the restriction $\pi_{m_n}\vert_{C_n'}$ satisfies $f_n\circ \pi_{m_n}\vert_{C_n'}\, \nu_C = \mu _n$. Moreover, if $n\leq p$, then $\iota_n(\dom f_n)\subseteq \iota_p(\dom f_p)$. so the family $\{ \iota_n(f_n)\mid n\geq 0\}$ is an increasing sequence of intervals in $C$, each of which is clopen (being the embedded image of a subset of an antichain of compact elements in $CT$). If we let $C' = \bigcup_n C_n'$, then $C'$ is an open, hence Borel, subset of $C$. 

We define $f_\mu\colon C'\to D$ by $f(c) = \sup_n f_n\circ \pi_{n}(c)$. \medbreak
\noindent\textbf{Claim:} $f_\mu$ is well-defined and measurable.\\[1ex]
\noindent\emph{Proof:} 
If $n\leq p$ then $f_n\circ \pi_{n,p}\leq f_p$, so $f_n\circ \pi_n = f_n\circ (\pi_{n,p}\circ \pi_p) = (f_n\circ \pi_{n,p})\circ \pi_p \leq f_p\circ\pi_p$. So we conclude that $\{ f_n\circ \pi_{n}(c)\mid n > 0\ \&\ c\in C_n'\}$ is an increasing sequence in $B_D$. Since $D$ is a domain, this sequence has a well-defined supremum, so $f_\mu(c) = \sup_{c\in C_n'} f_n\circ \pi_{n}(c)$ is well-defined for all $c\in C'$. 

To show measurability, it is enough to show $f_\mu^{-1}(X)$ is a Borel subset of $C'$ for all Scott-closed subsets $X\subseteq D$. If $X$ is such a set, then $c\in f_\mu^{-1}(X)$ iff $f_\mu(c) \in X$, and since $X$ is Scott-closed, this holds iff $f_n\circ \pi_{n}(c) \in X$ for all $n$ for which $c\in C_n'$. Since $n \leq p$ implies $C_n'\subseteq C_p'$, we have that $f_\mu^{-1}(X) = \liminf_n (f_n\circ \pi_{n})^{-1}(X) = \bigcup_n\bigcap_{p\geq n} (f_p\circ \pi_{p})^{-1}(X)$, is a countable union of sets, each of which is the intersection of closed subsets of $C'$, since $f_n$ and $\pi_n$ are relatively Scott-continuous. So, $f_\mu^{-1}(X)$ is an $F_\sigma$ subset of $C'$.
\medbreak
For the final claim, we appeal to Theorem~\ref{thm:topvsweak} (i): since $f = \sup_n f_n\circ \pi_n$, if $U\subseteq D$ is Scott open, then $f\, \nu_C (U) = \nu_C(f^{-1}(U))$. But $f^{-1}(U) = \bigcap_n (f_n\circ\pi_n)^{-1}(U)$, from which it follows that $\limsup_n f_n\circ\pi_n\, \nu_C(U) \geq f\,\nu_C(U)$. Hence Theorem~\ref{thm:topvsweak} (i)  implies $f\,\nu_C =  \sup_n f_n\circ\pi_n\,\nu_C = \sup_n f_n\,\nu_n = \sup_n \mu_n = \mu$.  
\end{proof}

\begin{corollary}\label{cor:skorohod}
If $D$ be a countably-based coherent domain, then for each $\mu\in  \textsf{SProb}\, D$ there is a measurable partial map $f_\mu\colon C\rightharpoonup D$ satisfying $f_\mu\, \nu_C = \mu$. 

Moreover, 
if $\mu_n\in\textsf{SProb}\, D$ with $\mu_n\to_w \mu\in \textsf{SProb}$ in the weak topology, then 
$ f_{\mu_n}\to f$ a.s. on $\dom f_\mu$ wrt $\nu_C$. 
\end{corollary}
\begin{proof}
Given $\mu\in \textsf{SProb}\, D$, we apply Proposition~\ref{prop:dommap1} to choose a sequence of simple measures $\cdots \ll \sigma_{\mu,m}\ll\sigma_{\mu,{m+1}}\ll \cdots \ll \mu$ with $\mu = \sup_m \sigma_{\mu,m}$, antichains $C_{k_m}\subseteq CT$ and partial maps $f_{\mu,m}\colon C_{k_m}\rightharpoonup B_D$ so that $f_{\mu,m}\, \nu_{k_m} = \sigma_{\mu,m}$ for each $m$ and $m\leq m'$ implies $f_{\mu,m}\leq f_{\mu,m'}$. We then define $f_\mu\colon C\rightharpoonup D$ by $f_\mu = \sup_m f_{\mu,m}$.  Theorem~\ref{thm:mapping} then implies $f_\mu\, \nu_C = \mu$. 

Likewise, if $\mu_n\to_w\mu$, then $\mu_n = \lim_{m'} \sigma_{\mu_n,m'}$ with $\sigma_{\mu_n,m'}\ll  \mu_{m_n}$ and $\sigma_{\mu,m'} = f_{\mu_{m'},k_{m'}}\, \nu_{k_m}$ for some $f_{\mu_n,m'}\colon C_{k_{m'}}\to D$, and $f_{\mu_n} = \sup_{m'} f_{\mu_n,m'}$ for each $n$.
The proof is complete if we show that $\mu_n\to_w \mu\in \textsf{SProb}\, D$ implies  $f_{\mu_n}\to_\Lambda f_\mu$ a.s.\ on $\dom f_\mu$ wrt $\nu_C$, where $\to_\Lambda$ denotes convergence wrt the Lawson topology. In fact, we show that $\{ c\in \dom f_\mu\mid f_{\mu_n}(c)\not\to_\Lambda f_\mu(c)\}$ is a $\nu$-null set. 

To begin, we note that for each $m > 0$, Theorem 3.7 of~\cite{misl-anat1} implies there is a canonical projection $\pi_m \colon CT \to \da C_m$ that is Scott continuous. From this it follows that $f_{\mu,m}\circ \pi_m\colon \ua C_m\to D$ is Scott continuous on the subdomain $\ua C_m\subseteq CT$. \\[1ex]
\emph{$1^\circ:$ We claim for each $m > 0$ there is some $N_m > 0$ with $f_{\mu,m}\circ \pi_{k_m} \leq f_{\mu_n}$ for $n\geq N_m$.}
Indeed, given $m$, since $\sigma_{\mu,m}\ll \mu = \lim_n \mu_n$,  there is some $N_m > 0$ so that $\sigma_{\mu,m}\ll \mu_n$ for $n\geq N_m$. And, since $\mu_n = \sup_{m'} \sigma_{\mu_n,m'}$ and $\sigma_{\mu_n,m'}\ll \mu_n$, it follows that for each $n \geq N_m$, there is some $m_n$ with $\sigma_{\mu,m}\ll \sigma_{\mu_n,m'}$ for each $m'\geq m_n$. Now, we have $f_{\mu,m}\colon C_{k_m}\to D$ and $f_{\mu_n,m'}\colon C_{k_{m'}}\to D$ satisfying $f_{\mu,m}\, \nu_{k_m} = \sigma_{\mu,m}$ and $f_{\mu_n,m'}\, \nu_{k_{m'}} = \sigma_{\mu_n,m'}$. By increasing $m'$ if needed, we also can assume that $k_{m'} > k_m$. Then $\sigma_{\mu,m}\ll \sigma_{\mu_n,m'}$  implies that $f_{\mu,m}\circ \pi_{k_m,k_m'}\leq f_{\mu_n,m'}$ by the Splitting Lemma~\ref{thm:split}. Since $f_{\mu_n} = \sup_{m'} f_{\mu_{k_n},m'}$, it follows that $f_{\mu,m}\circ \pi_{k_m}\leq f_{\mu_n}$ for each $n\geq N_m$, as claimed.\\[1ex]
\indent 
Since $D$ has a countable base, $D$ has a countable basis of Lawson-open sets of the form $O = U\setminus\ua F$, where $U$ is Scott open and $F$ is finite. For each such open set $O$, let\\[1ex]
\centerline{$ A_O = \{ c\in \dom f_\mu\mid f_\mu(c)\in O\ \&\ (\forall n)(\exists k_n\geq n)\, f_{\mu_{k_n}}(c)\not\in O\}$.}\\[1ex]  
\emph{$2^\circ:$ We show $f_{\mu_{k_n}}(A_O)\subseteq \ua F$:} Given $c\in A_O$ and $n >0$, we claim there is some $k_n\geq n$ with $f_{\mu_{k_n}}(c)\in \ua O\setminus O$. Indeed, $\mu = \sup_m \sigma_{\mu,m}$ and $f_\mu(c)\in O$ implies there is some $m$ with $f_{\mu,m}(c)\in O$. Then,  $\sigma_{\mu,m}\ll \mu$ implies there is some $k > n, m$ with $\sigma_{\mu,m}\leq \mu_{k'}$ for all $k'\geq k$. Since $c\in A_O$, there is some $k_n\geq k$ with   $f_{\mu_{k_n}}(c)\not\in O$. Then $\sigma_{\mu,m}\leq \mu_{k_n}$ implies $f_{\mu,m}(c)\leq f_{\mu_{k_n}}(c)$, so $f_{\mu_{k_n}}(c)\in \ua f_{\mu,m}(c)\subseteq \ua O$.
Thus, $f_{\mu_{k_n}}(c)\in\ua O\setminus O$, and since $O = U\setminus \ua F$, it follows that $f_{\mu_{k_n}}(c)\in \ua F$. This shows $f_{\mu_{k_n}}(A_O)\subseteq \ua F$. \\[1ex]
\noindent\emph{$3^\circ:$ Now we show $\nu_C(A_O) = \nu_C(\ua F) = 0$:} Indeed $\mu_n\to_w \mu$ in the weak topology implies $\limsup_n \mu_n(\ua F)\leq \mu(\ua F)$ by Theorem~\ref{thm:topvsweak}. 
Since $ f_\mu(A_O) \subseteq O = U\setminus \ua F$, we have $\mu(\ua F) = 0$, and so $\mu_n(\ua F) = 0$ for all $n$. But $f_{\mu_{k_n}}(c)\in f_{\mu_{k_n}}(A_O)\cap \ua F$ then implies $\nu_C(\ua F) = 0$. \\[1ex]
\noindent\emph{$4^\circ:$ To conclude the proof,} we note that if $c\in \dom f_\mu$ and $f_{\mu_n}(c)\not\to_\Lambda f_\mu(c)$, then there is some Lawson open set $O = U\setminus \ua F$ with $f_\mu(c)\in O$ and $f_{\mu_n}(c)\not\in O$ for cofinally many $n$. We have shown that for such an open set $O$, we have $\nu_C(A_O) = 0$.  
Since there is a countable basis of such sets $O$, we conclude $\bigcup_O A_O = \{c\in \dom f\mid f_{\mu_n}(c)\not\to_\Lambda f_\mu(c)\}$ is a $\nu_C$-null set as well. 
\end{proof}
\section{Domains, Polish spaces and Skorohod's Theorem}\label{sec:polish}
In this section we present our main results. We begin with the necessary background about Polish spaces and random variables, and then develop a representation theorem for Polish spaces that have a topological embedding as a $G_\delta$-subset of a domain in the relative Scott topology. Once we have our general result, we focus in on probability measures and derive  a domain-theoretic proof of Skorohod's Theorem. Our starting point is the standard class of spaces studied using random variables. 
\begin{definition}
A \emph{Polish space} is a completely metrizable, separable topological space. I.e., $P$ is Polish iff $P$ is homeomorphic to a complete metric space that has a countable dense subset. 
\end{definition}

Polish spaces figure prominently in probability theory~\cite{billings}, as well as in descriptive set theory~\cite{kechr}. As we commented in the last section, one approach to probability theory~\cite{billings} begins with a metric space, and Polish spaces, are a common place where the deepest results hold. Since our results all involve separable spaces, the measurable sets in all cases are Borel sets. So from here on, we use the term Borel set, instead of measurable set.

The appropriate mappings in probability theory are \emph{random variables} -- measurable maps $X\colon (P,\Sigma_P, \mu)\to (S,\Sigma_S)$, where $(P,\Sigma_P,\mu)$ is a probability space, $(S,\Sigma_S)$ is a measure space, and $X^{-1}(A)\in\Sigma_P$ for each $A\in \Sigma_S$. If $X\colon P\to S$ is a random variable, then the \emph{push forward} of $\mu$ by $X$ is the measure $X\,\mu$ that satisfies $X\,\mu (A) = \mu(X^{-1}(A))$ for all measurable sets $A\subseteq S$. Equivalently, a function $f\colon S\to \Re$ is $X\,\mu$-integrable iff $f\circ X$ is $\mu$-integrable, and in this case, $\int fd X\,\mu = \int f\circ X d\mu$. We now show how our results from Section~\ref{sec:cantor} can be applied to Polish spaces that can be topologically embedded as the maximal elements of a domain in the relative Scott topology. 
We begin by extending some results about probability measures on Polish spaces to sub-probability measures. 

\begin{definition} Let $\mu$ be a probability measure on a space $X$. A Borel set $A\subseteq X$ is a \emph{$\mu$-continuity set} if $\mu(\overline{A}\setminus A) = 0$. Since $\overline{A}$ is closed and $A$ is Borel, $\overline{A}\setminus A$ is Borel.
\end{definition}

\begin{theorem}\label{thm:port}(Portmanteau Theorem)
Let $X$ be a Polish space, and let $\textsf{SProb}\, X$ denote the family of sub-probability measures on $X$ in the weak topology, and let $\mu_n,\mu\in\textsf{SProb}\, X$. Then the following are equivalent:
\begin{enumerate}
\item $\mu_n\to_w \mu$ in the weak topology.
\item $\int f d\mu_n \to \int f d\mu$ for all bounded uniformly continuous $f\colon X\to \Re$. 
\item $\liminf_n \mu_n(O)\geq \mu(O)$ for all open sets $O\subseteq X$.
\item $\limsup_n \mu_n(F)\leq \mu(F)$ for all closed sets $F\subseteq X$.
\item $\mu_n(A) \to \mu(A)$ for all $\mu$-continuity sets $A$. 
\end{enumerate}
\end{theorem}
\begin{proof}
Theorem 2.1 of~\cite{billings} shows these conditions are equivalent if $\mu_n,\mu$ are probability measures on a metric space. We assume the metric $d$ on $X$ satisfies $\text{diam}\, X < 1$ by normalization, if necessary, and then create a new space $X' = X\stackrel{\cdot}{\cup} \{*\}$, where $*$ is an element not in $X$. If $d$ is the metric on $X$, we extend $d$ to $X'$ by setting $d(*,x) = d(x,*) = 1$ for all $x\in X$. This makes $X'$ into a Polish space, and there is an embedding $e\colon \textsf{SProb}\, X \hookrightarrow \textsf{Prob}\, X'$ by $e(\mu) = \mu + (1 - \mu(X))\delta_*$. Then the conditions (i) -- (v) are equivalent for $e(\textsf{SProb}\, X)$.  But $\mu = e(\mu)\vert_X$ for all $\mu\in \text{SProb}\, X$, and $X$ is clopen in $X'$, so they also are equivalent for $\textsf{SProb}\, X$. 
\end{proof}

For the next result, recall that a measure $\mu$ on a space $X$ is \emph{concentrated on the Borel set $A\subseteq X$} if $\mu(X\setminus A) = 0$. 
\begin{proposition}\label{prop:prob-embed}
Let $X$ be a Polish space with a topological embedding $\iota\colon X\to D$ into a countably-based domain as a $G_\delta$-subset. Define $e\colon \textsf{SProb}\, X\to \textsf{SProb}\, D$ by $e(\mu) = \iota\,\mu$. Then:
\begin{enumerate}
\item $e$ is one-to-one.
\item $e(\textsf{SProb}\, X) = \{ \mu\in \textsf{SProb}\, D\mid \mu \text{ concentrated on } \iota(X)\}$, and $j\colon e(\textsf{SProb}\, X)\to \textsf{SProb}\, X$ by $j(\nu) = \nu \circ \iota$ is inverse to $e$. 
\item $e(\textsf{Prob}\, X) \subseteq \textsf{Prob}\, D = \text{Max}\, \textsf{SProb}\, D$ is a Borel subset of $\textsf{SProb}\, D$. 
\end{enumerate}
\end{proposition}
\begin{proof} The same result for $\textsf{Prob}\,X$ is Proposition 4.2 of ~\cite{weaktop}. That proof relies on Proposition 4.1 of the same paper, which characterizes properties of the embedding of $X$ into $D$, and hence applies equally to sub-probability measures.  With this result in hand, the proofs in~\cite{weaktop} apply almost verbatim for sub-probability measures, except that part (ii) uses $\mu(X) = 1$, but this can be replaced by $\mu(X) = ||\mu||$. The final point that $e(\textsf{SProb}\, X)$ is a Borel set follows from Lemma 2.3 of~\cite{varad}.
\end{proof} 

\begin{theorem}\label{thm:embedcont}
Let $X$ be a Polish space and $\iota\colon X\to D$ a topological embedding of $X$ into a countably based domain $D$. Then the mapping $e\colon \textsf{SProb}\, X\to \textsf{SProb}\, D$ by $e(\mu) = \iota\, \mu$ is a topological embedding relative to the weak topologies on $\textsf{SProb}\, X$ and on $\textsf{SProb}\, D$. 
\end{theorem}
\begin{proof}
We first show $e\colon \textsf{SProb}\, X\to \textsf{SProb}\, D$ is continuous: Let $\mu_n,\mu\in\textsf{SProb}\, X$ with $\mu_n\to_w \mu$. To show $e(\mu_n)\to_w e(\mu)$, let $O\subseteq D$ be Lawson open. Then:
\begin{eqnarray*}
\liminf_n e(\mu_n)(O) &= & \liminf_n \iota\, \mu_n(O)\\ &=& \liminf_n \mu_n(\iota^{-1}(O)) \notag\\
& \geq & \mu(\iota^{-1}(O)) \qquad\qquad \text{by Theorem~\ref{thm:port}(ii)}\\
&=& e(\mu)(O).
\end{eqnarray*}
Then Theorem~\ref{thm:topvsweak}(ii) implies $e(\mu_n)\to_w e(\mu)$. 

For the converse, let $\nu_n\to_w \nu$ in $\textsf{SProb}\, D$ with $\nu_n, \nu$ all concentrated on $\ran e$. Since $\iota\colon X\to D$ is a topological embedding, given an open set $O\subseteq X$, there is some Lawson open $O'\subseteq D$ satisfying $\iota(O) = O'\cap \iota(X)$. Then,
\begin{eqnarray*}
\liminf_n j(\nu_n)(O) & = & \liminf_n \nu_n(\iota(O))\\  &=&  \liminf_n \nu(\iota(O)) \\
& = & \liminf_n \nu_n(O'\cap \iota(X))\\
& = & \liminf_n \nu_n(O') \qquad\quad\ \text{$\nu$ is concentrated on $\iota(X)$}\\
&\geq & \nu(O') \qquad\qquad\qquad\ \text{ by Theorem~\ref{thm:topvsweak}(ii)}\\
&=& j(\nu)(O).
\end{eqnarray*}
It follows by Theorem~\ref{thm:port}(ii) that $j(\nu_n)\to_w j(\nu)$ in \textsf{SProb}$\, X$.
\end{proof}
\begin{remark}
Theorem~\ref{thm:embedcont} is Corollary 4.1 of~\cite{weaktop} extended to the case of sub-probability measures, from the case of probability measures. The proof is identical to the one in~\cite{weaktop}, except the reasoning has been changed to include the results we established for \textsf{SProb}.
\end{remark}

\begin{theorem}\label{thm:cantskor}
Let $X$ be a Polish space. For each $\mu\in \textsf{SProb}\, X$ there is a measurable partial map $f_\mu\colon C\rightharpoonup X$ satisfying $f_\mu\, \nu_C = \mu$. Moreover, if $\mu_n, \mu\in \textsf{SProb}\, X$ and $\mu_n\to_w \mu$, then $f_{\mu_n}\to f_\mu$ a.s. on $\dom f_\mu$ wrt $\nu_C$.
\end{theorem}
\begin{proof}
This follows from Corollary~\ref{cor:skorohod}, Theorem~\ref{thm:polish} and Theorem~\ref{thm:embedcont}
\end{proof}
\subsection{Bounded complete domains and Skorohod's Theorem}
The connection between domains and Polish spaces involves computational models. A \emph{computational model for a topological space $X$} is a domain $D$ for which there is a topological embedding $X\simeq \text{Max}\, D$ of $X$ as the space of maximal elements of $D$ endowed with the relative Scott topology. As described in the Introduction, this notion emerged from the work of Edalat, who developed the first domain models of spaces arising in real analysis using the domain of compact subsets of the space under reverse inclusion. Later, Lawson~\cite{lawson} showed that the space $\text{Max}\, D$ of maximal elements of a bounded complete countably based domain in the relative Scott topology is a Polish space, and Ciesielski, Flagg and Kopperman~\cite{kopperman,kopperman2} showed that every Polish space has such a model. Finally, Martin~\cite{martin} noted that the space of maximal elements of any countably based, bounded complete domain is a $G_\delta$ in the relative Scott topology. Since these results play a crucial role in our work, we state them formally:

\begin{theorem}\label{thm:polish}(Lawson~\cite{lawson}, Ciesielski, et al.~\cite{kopperman,kopperman2}, Martin~\cite{martin}) 
A space $X$ is representable as $\text{Max}\, D$ in the relative Scott topology, for $D$ a countably based, bounded complete domain, iff $X$ is a Polish space. In such a representation, $X$ is a $G_\delta$--subspace of $D$ in the Scott topology. 
\end{theorem}

Since bounded complete domains are coherent~\cite{abrjung}, our results from Section~\ref{sec:cantor} apply to them, but bounded completeness allows us to refine the mappings used.
Recall that $[CT\rightharpoonup D]$ denotes the family  of the partial maps $f\colon C\rightharpoonup D$ that are relatively Scott continuous on their domain, which is a subset of a Lawson closed antichain $C\subseteq CT$. Recall our notation~\ref{not:nota1} that if $f, g$ are such maps, then $f\leq g$ iff $\dom f\subseteq \da \dom g$ and $f\circ \pi_{\dom f,\dom g} \leq g$, where $\pi_{\dom f,\dom g}$ is defined. 

\begin{proposition}\label{prop:scottcont}
Let $D$ be a bounded complete domain, and let $f\colon C_n \hookrightarrow D$ be a partial map. Then $f$ has a Scott-continuous extension $\wid{f}\colon CT \to D$. Moreover, if $m\leq n$ and $f\colon C_m\rightharpoonup D$ and $g\colon C_n\rightharpoonup D$ are such maps with $f\leq g$, then $\wid{f}\leq \wid{g}$.
\end{proposition}

\begin{proof}
If $f\colon C_m\rightharpoonup D$, then we define $\wid{f}\colon CT\to D$ in two steps. First, we define\\[1ex] 
\centerline{$f_0\colon \da C_n \to D$ by $f_0(x) = \begin{cases} \inf f(\ua x\cap \dom f) & \text{ if } \ua x\cap C_n\subseteq  \dom f \\
\perp_D & \text{ otherwise}\end{cases}$.}\\[1ex] 
If $x\leq y$ and $x$ satisfies the first condition in the definition, then so does $y$. Then it's clear that $f_0$ is monotone and that $f_0(x) = f(x)$ for all $x\in \dom f$. Since $\da C_n\subseteq KCT$, $f_0$ is Scott-continuous. But then $\wid{f} = f_0\circ \pi_{\da C_n}\colon CT\to D$ is Scott continuous and clearly extends $f_0$, and hence also extends $f$. Finally, if $f\leq g$, then $f_0\leq g_0$ is easy to show, from which $\wid{f}\leq \wid{g}$ follows. 
\end{proof}

\begin{corollary}\label{cor:skorohodcont}
If $D$ be a countably-based bounded complete domain, then for each $\mu\in  \textsf{SProb}\, D$ there is a Scott-continuous map $f_\mu\colon CT\to D$ satisfying $f_\mu\vert_C\, \nu_C = \mu$. 

Moreover, 
if $\mu_n\in\textsf{SProb}\, D$ with $\mu_n\to_w \mu\in \textsf{SProb}$ in the weak topology, then 
$ f_{\mu_n}\vert_C\to f\vert_C$ a.s. on $C$ wrt $\nu_C$. 
\end{corollary}

\begin{proof}
Combine Theorem~\ref{thm:cantskor} and Proposition~\ref{prop:scottcont}. 
\end{proof}

\begin{example}\label{exam:one}
One might hope that with Scott-continuous maps $f_{\mu_n}, f_\mu\colon CT\to D$, a stronger conclusion about the convergence $f_{\mu_n}\vert_C\to f_\mu\vert_C$ could be derived. But measurability of the functions $f_{\mu_n}\vert_C, f_\mu\vert_C$ is the most to be expected, so a.s.\ convergence also is the best one can do. To illustrate, consider the sequence of probability measures $\mu_n = {2^n - 1\over 2^n}\delta_0 + {1\over 2^n}\delta_1\in \textsf{Prob}\, (\{0,1\}_\perp)$ on the lifted two-point flat poset $\{0,1\}$. Then $\mu_n \to_w \mu = \delta_0$, and the construction using approximants $\sigma_{\mu,m} = {2^n-1\over 2^n}\delta_0 + {1\over 2^n}\delta_\perp\ll \mu$  yield $f_\mu(1) =\, \perp$, while $f_\mu\vert_{C\setminus \{1\}} = 0$. So, while $f_\mu$ is Scott-continuous on $CT$, its restriction to $C$ is not even lower semicontinuous. As we'll see in Subsection~4.1 below, upper semicontinuity is another matter. 
\end{example}
A \emph{stochastic process on a measure space $(S,\Sigma_S)$} is a family $\{ X_t\mid t\in T\subseteq \Re+\}$ of random variables $X_t\colon (X,\Sigma_X,\nu)\to (S,\Sigma_S)$, where $(X,\Sigma_X, \nu)$  is a probability space. It's often assumed that $(S,\Sigma_S)$ is a Polish space, in which case $\textsf{Prob}\, S$ also is Polish in the weak topology. The push forward measure $X_t\, \nu\in \textsf{Prob}\, S$ is called the \emph{law of $X_t$}, and  a natural question is the convergence properties of the family $\{X_t\, \nu\mid t\in T\}\subseteq \textsf{Prob}\, S$. Since $\textsf{Prob}\, S$ is Polish, convergence can be defined using sequences, and it's obvious that if $X_{t_n}\to X_t$ a.e. on $X$, then $X_{t_n}\, \nu \to_w X_t\, \nu$ in $\textsf{Prob}\, S$. Skorohod's Theorem not only provides a converse to this observation, it also shows the probability space $(X,\Sigma_X,\nu)$ can be assumed to be the unit interval with Lebesgue measure:

\begin{theorem}\label{thm:skor} (Skorohod's Theorem~\cite{skorohod}) 
If $P$ is a Polish space and $\mu\in\textsf{Prob}\, P$, then there is a random variable $X\colon [0,1]\to P$ satisfying $X\,\lambda = \mu$. 

Moreover, if $\mu_n, \mu\in \textsf{Prob}\,P$ satisfy $\mu_n\to_w \mu$ in the weak topology, then the random variables $X,X_n\colon [0,1]\to P$ satisfying $X\,\lambda = \mu$ and $X_n\,\lambda = \mu_n$, also satisfy $X_n\to X$ a.s.~wrt Lebesgue measure. 
\end{theorem} 

Before we prove the theorem, we need one more preparatory result. W.\ M.\ Schmidt was the first to observe that the canonical surjection $\pi\colon C\to [0,1]$ from the Cantor set, $C\simeq 2^\Nat$ sends Haar measure $\nu_C$ to Lebesgue measure. We need the lower adjoint of that projection for our proof.

\begin{proposition}\label{prop:lebesgue} 
Let $C\simeq 2^\Nat$ denote the Cantor set regarded as the countable product of two-point groups. Then the canonical projection $\pi\colon C\to [0,1]$ has a lower adjoint $j\colon [0,1]$ satisfying $j([0,1]) = C\setminus KC$, and $j\, \lambda = \nu_C$,  where $\lambda$ denotes Lebesgue measure.
\end{proposition}
\begin{proof}
For a proof that $\pi\, \nu_C = \lambda$, see~\cite{brian-misl}, which also has an extensive discussion of related results, including the following. First, $\pi$ preserves all sups and all infs, so, in particular, it has a lower adjoint $j\colon [0,1]\to C$ preserving all suprema. Then $j([0,1]) = C\setminus KC$, where $KC$ is the set of compact elements, which is countable, so $\nu_C(KC) = 0$. The pair of maps, $\pi\vert_{C\setminus KC}$ and $j$, form a Borel isomorphism. 

We claim $j\, \lambda = \nu_C$: Indeed, if $A\subseteq C$ is a Borel set, then 
\begin{eqnarray*}
j\, \lambda(A) = \lambda(j^{-1}(A)) &=& \lambda(j^{-1}(A\setminus KC))\\ &=&\pi\, \nu_C(j^{-1}(A\setminus KC))\\ &=& \nu_C( \pi^{-1}\circ j^{-1} (A\setminus KC))\\ &=&  \nu_C(A\setminus KC) = \nu_C(A),
\end{eqnarray*} 
since $\nu_C(KC) = 0$.
\end{proof}

\begin{proof} (of Theorem~\ref{thm:skor}) 
We can restrict the results in Corollary~\ref{cor:skorohodcont} to $\textsf{Prob}\, P$, in which case the mapping $f_\mu\colon C\to S$ is total for each $\mu\in\textsf{Prob}\ P$. Given $\mu_n,\mu\in \textsf{Prob}\, P$ with $\mu_n\to_w \mu$, then $f_{\mu_n}\to f_\mu$ a.s.~on $C$ wrt $\mu_C$. If we precompose these mappings with $j$, then we have random variables $X_n = f_{\mu_n} \circ j, X = j\circ f_\mu\colon [0,1]\to P$ as desired: $X_n\, \lambda = f_{\mu_n}\, \nu_C = \mu_n, X\, \lambda = f_\mu\,\nu_C = \mu$, and $X_n\to X$ a.s. on $[0,1]$ wrt $\lambda$. 
\end{proof}

\begin{remark} A \emph{standard Borel space} is a measurable space for which there is a Borel isomorphism onto a Borel subset of a Polish space~\cite{kechr}. This leads to two comments:
\begin{enumerate}
\item The thrust of Skorohod's Theorem is that the domain for any stochastic process whose range is a Polish space can be assumed to be a standard Borel space. Some statements of the theorem simply state it that way, without specifying which standard Borel space is being used. But most often, the standard space is assumed to be the unit interval with Lebesgue measure.
\item  Any two standard Borel spaces that are uncountable are Borel isomorphic (cf.~\cite{kechr}). Clearly $C$ is such a space, as is the canonical example, $[0,1]$. So we could have used $C$ in Theorem~\ref{thm:skor} instead of $[0,1]$.
\end{enumerate}
\end{remark}

\subsection{The special case of chains}\label{subsec:chain}
Theorem~\ref{thm:polish} asserts that Polish spaces are exactly those for which there is a computational model in a bounded complete domain. And we just found that when $D$ is bounded complete, the mappings $f_\mu\colon CT\to D$ are Scott continuous. But Example~\ref{exam:one} shows that the mappings $f_\mu\vert_C$ are still only measurable. So it's natural to ask when is it possible to assure the mappings $f_\mu\vert_C$ are lower- or upper semicontinuous, rather than just measurable. We do not have a general answer, but there is one case where we can say something, namely when the Polish space $X$ is endowed with a closed total order. The proof technique is different -- rather than an indirect approach using the Cantor tree, we take a direct approach using the unit interval.
\begin{notation}
Throughout this section, \textbf{we assume $D$ is a chain.}  
\end{notation}
 
 \begin{definition}\label{def:cmd}
 Let $\mu$ be a sub-probability measure on $D$. The \emph{cumulative distribution function} $F_\mu\colon D\to [0,1]$ is defined by $F_\mu(x) = \mu(\da x)$.  
 \end{definition}
 
 \begin{proposition}\label{prop:cmd}
 For each $\mu\in \SPr{D}$, $F_\mu$ preserves all infima. 
 \end{proposition}
 \begin{proof}
 Let $\mu$ be a sub-probability measure. If $x\leq y\in D$, then $\da x\subseteq \da y$, so $F_\mu(x) = \mu(\da x) \leq \mu(\da y) = F_\mu(y)$. So $F_\mu$ is monotone, and since $D$ is a chain, this means $F_\mu$ also preserves finite infima. Now, any filtered set $A\subseteq D$ is totally ordered  because $D$ is. Then $\da \inf A = \bigcap_{x\in A} \da x$, and so\\[1ex]
 \centerline{$F_\mu(\da \inf A) = F_\mu(\bigcap_{x\in A} \da x) = \mu(\bigcap_{x\in A} \da x) = \inf_{x\in A} \mu(\da x) = \inf_{x\in A} F_\mu(x)$,}\\[1ex]
 where the next-to-last equality follows from the fact that, being a  Scott-continuous valuation on $D$, $\mu$ preserves directed unions of open sets, so it preserves filtered intersections of closed sets, such as $\{ \da x\mid x\in A\}$. This shows $F_\mu$ also preserves filtered infima, and so it preserves all infima.
 \end{proof}
 
Since $F_\mu$ preserves all infima and $D$ is a continuous lattice, it follows that $F_\mu$ is an upper adjoint, so it has a unique lower adjoint $G_\mu \colon [0,1]\to D$ defined by $G_\mu(r) = \inf F_\mu^{-1}(\ua r)$. We denote this relationship by $F_\mu \dashv G_\mu$. 

We recall some facts about such adjoint pairs; for more detail, see Chapter 0 of~\cite{comp}. 
First, each component of an adjoint pair $f\colon L \to M$, $g\colon M\to L$ with $f\dashv g$ determines the other. The formula for $G$ above shows how to define the lower adjoint, given an upper adjoint: $g(x) = \inf f^{-1}(\ua x)$. Conversely, given a lower adjoint $g$, the upper adjoint $f$ is given by $f(y) = \sup g^{-1}(\da y)$. Upper adjoints preserve all infima, and lower adjoints preserve all suprema. Moreover, if $f\dashv g$ and $f'\dashv g'$, then $f\leq f'$ iff $g'\geq g$. Finally, the components earn their names because of the relationship $f\circ g\geq 1_M$ and $g\circ f\leq 1_L$.
 
\begin{proposition}\label{prop:uuperadj}
 If $\mu$ is a sub-probability measure on $D$ with cumulative distribution function $F_\mu$, then the upper adjoint, $G_\mu\colon [0,1]\to D$ satisfies $G_\mu\, \lambda = \mu$, where $\lambda$ denotes Lebesgue measure. 
\end{proposition}
\begin{proof}
If $x\in D$, then 
\begin{eqnarray*}
G_\mu\, \lambda(\da x) & = & \lambda(G_\mu^{-1}(\da x))\\
& = & \lambda(\da F_\mu(x)) \qquad\qquad \qquad F_\mu\dashv G_\mu\\
& = & F_\mu(x) = \mu(\da x)
\end{eqnarray*}
Since $G_\mu\,\lambda$ and $\mu$ agree on Scott-closed sets, it follows that $G_\mu\, \lambda = \mu.$
\end{proof} 
 
 \begin{theorem}\label{thm:chain}
 If $D$ is a chain and $KD = \{\perp\}$, then $G\mapsto G\, \lambda\colon [[0,1]\to D]\to \SPr{D}$ is an order-isomorphism. Therefore, $\SPr{D}$ is a continuous lattice, and the same is true of $\Pr{D}$ is a domain.
 \end{theorem}
 \begin{proof}
 Each Scott-continuous map $G\colon [0,1]\to D$ preserves all suprema, since the domain $D$ is a chain. And each such map determines a sub-probability measure $G\, \lambda$. Then the cumulative distribution $F_{G\,\lambda}\colon D\to [0,1]$ satisfies $F_{G\,\lambda}(x) = G\,\lambda(\da x) = \lambda(G^{-1}(\da x)) = \sup G^{-1}(\da x)$. This means $F_{G\,\lambda}$ is the upper adjoint of $G$. Since upper and lower adjoints uniquely determine one another, the mapping $G\mapsto G\,\lambda$ has an inverse sending $\mu$ to the lower adjoint of $F_\mu$. 
 
 For the order structure, suppose $G\leq G'\in [[0,1]\to D]$. We show $G\,\lambda\leq G'\,\lambda$: Then given $x\in D$ and $r\in [0,1]$, if $G'(r) \leq x$, then $G(r)\leq x$; said another way, $G'^{-1}(\da x) \subseteq G^{-1}(\da x)$, so \\[1ex]
 \centerline{$G'\,\lambda(\da x) = \lambda(G'^{-1}(\da x)) = \sup G'^{-1}(\da x) \leq \sup G^{-1}(\da x) = \lambda(G^{-1}(\da x)) = G\,\lambda(\da x)$.}\\[1ex]
If $x =\, \perp$, $G\,\lambda(\Ua x) = G\,\lambda (D) \leq G'\,\lambda(D) = G'\,\lambda(\Ua x)$. On the other hand, since $KD = \{\perp\}$, then $x > \perp$ implies  $D = \da x\stackrel{\cdot}{\cup} \Ua x$, so we have\\[1ex]
 \centerline{$G\,\lambda(\Ua x) = G\,\lambda(D) - G\,\lambda(\da x) \leq G\lambda(D) - G'\,\lambda(\da x)\leq G'\lambda(D) - G'\lambda(\Ua x)  = G'\,\lambda(\Ua x)$.}\\[1ex] 
Since $D$ is a chain, every Scott-open set has the form $\Ua x$ for some $x\in D$, so $G\,\lambda\leq G'\,\lambda$. 

Conversely, if $\mu\leq \nu$, then $\mu(\da x)\geq \nu(\da x)$ by the same argument we used above,  so $F_\mu(x) = \sup \mu(\da x) \geq \sup  \nu(\da x) = F_\nu(x)$. It follows that $G_\mu\leq G_\nu$ from our remarks about adjoint pairs. 

Thus, the correspondence $G\mapsto G\,\lambda$ is an order-isomorphism. Since $[0,1]$ and $D$ are continuous lattices, they are both bounded complete domains, so $[[0,1]\to D]$ is a bounded complete domain. But $x\mapsto \top$ is the largest element of $[[0,1]\to D]$, so this is a continuous lattice.   It follows that $\SPr{D}$ is a continuous lattice as well. 

For the final claim, the mapping $\mu\mapsto \mu + (1 - \mu(D))\delta_\perp: \SPr{D}\to \Pr{D}$ is a closure operator that preserves directed sups, and the image of a continuous lattice under such a closure operator is a continuous lattice (cf.~\cite{comp}, Definition 0-2.10ff.).
 \end{proof}
\section{Summary and Future Work}
In this paper we have developed a domain-theoretic approach to random variables. Our main results have proved some standard results in probability theory using the Cantor tree as a domain, in which measurable (partial) maps from the Cantor set are approximated by Scott-continuous (partial) maps defined on subsets of the tree. We have obtained extensions of Skorohod's Theorem to the case of sub-probabilities, both for domains and for the classic case of Polish spaces. In the latter case, we have shown that the approximating maps on the Cantor tree are Scott-continuous. A novel aspect of our approach is the approximation of measurable (partial) maps from the Cantor set to a domain by Scott-continuous (partial) maps defined on the underlying tree. The novel proof-theoretic technique was the use of the transport numbers from the Splitting Lemma to show how to define the approximating partial maps. Finally, we have given a direct proof that the sub-probability measures and the probability measures on a complete chain form a continuous lattice. This result offers the first new insight in over two decades to the domain structure of $\SPr{D}$ and $\Pr{D}$, the last such results having appeared in~\cite{jungtix}.

There are a number of interesting questions to be explored. Finding further results from random variables that can be obtained using the techniques presented here is an obvious issue. Another question we have been exploring is the potential use of the disintegration theory for product measures, in order to understand the domain structure of $\SPr{(D\times E)}$, in the case $D$ and $E$ are chains. In particular, we'd like to know if we can use the fact that  $\SPr{D}$ and $\SPr{E}$ are continuous lattices to derive some insight into the domain structure of $\SPr{(D\times E)}$. Our first idea -- that this family of measures also would be a lattice -- is debunked by a simple example in~\cite{jones}, so more subtle issues are at play here. Last, we're interested in investigating the potential application of our ideas to the setting of quantum computation and quantum information; in particular, is there a role for this approach if one regards measurements as random variables?
\begin{acknowledgement}
The author wishes to acknowledge the support of the US AFOSR during the preparation of this work.
\end{acknowledgement}
\bibliographystyle{plain}

\end{document}

%% file: macros.tex
\def\.{{\cdot}}            
\def\<{\langle}            
\def\>{\rangle}            
\def\({\big(}              
\def\){\big)}              
\def\implies
{\hbox{$\Rightarrow$}}     
\def\lead{\leaders\hbox to 1.5ex{\hss${.}$\hss}\hfill}
\def\arr{\hbox to 60pt{\rightarrowfill}}
\def\larr{\hbox to 60pt{\leftarrowfill}}

\newskip\aline \newskip\halfaline
\def\1{{\bf1}}

\renewcommand{\Re}{{\mathbb R}}
\newcommand{\Nat}{{\mathbb N}}
\renewcommand{\P}{\mathcal{P}}

\newcommand{\B}{{\mathcal{B}}}

\newcommand{\ua}{{\uparrow}}
\newcommand{\da}{{\downarrow}}

\newcommand{\DCPO}{\mbox{\sf DCPO }}

 \newcommand{\DOM}{\mbox{\textsf{DOM} }} 
  \newcommand{\SDOM}{\mbox{\textsf{SDOM} }} 
 \newcommand{\BCD}{\mbox{\textsf{BCD} }} 
 \newcommand{\RB}{\mbox{\textsf{RB} }}

\renewcommand{\O}{{\mathcal O}}

\newcommand{\M}{{\mathcal M}}

\newcommand{\Da}{\mathord{\mbox{\makebox[0pt][l]{\raisebox{-.4ex}
                           {$\downarrow$}}$\downarrow$}}}
\newcommand{\Ua}{\mathord{\mbox{\makebox[0pt][l]{\raisebox{.4ex}
                           {$\uparrow$}}$\uparrow$}}}
\renewcommand{\Pr}{\textsf{Prob}}